\documentclass[fleqn, 11pt]{amsart}

\usepackage{hyperref}
\usepackage{amssymb}
\usepackage{amsfonts}
\usepackage[english]{babel}
\usepackage{euscript}
\usepackage{bbm}
\usepackage{xfrac}
\usepackage{color}

\usepackage{color}
\usepackage{mathrsfs}
\usepackage{calrsfs}
\usepackage{amsbsy}
\let\mathcal\EuScript
\allowdisplaybreaks[4]

\newcommand{\fl}{\hspace*{-\mathindent}}

\usepackage[all]{xy}
\SelectTips{cm}{}

\newcommand{\PDE}{\textsc{pde}}

\newcounter{example_counter}
 
\newcounter{definition_counter}
\newcounter{theorem_counter}
\newcounter{remark_counter}
\newcounter{proposition_counter}
\newcounter{lemma_counter}
\theoremstyle{theorem}
\newtheorem{proposition}{\textsc{Proposition}}

\newtheorem{lemma}{\textsc{Lemma}}
\theoremstyle{definition}

\theoremstyle{remark}
\newtheorem{remark}{\textsc{Remark}}

\DeclareMathOperator{\Der}{Der}
\DeclareMathOperator{\rank}{rank}
\DeclareMathOperator{\dC}{\mathbf{X}}
\DeclareMathOperator{\sym}{sym}
\DeclareMathOperator{\Ev}{\mathbf{E}}
\newcommand{\ldb}{[\![}
\newcommand{\rdb}{]\!]}

\let\phi=\varphi

\begin{document}
\title[Lagrangian extensions of multi-dimensional
equations.]{Lagrangian extensions of multi-dimensional integrable
  equations. I. The five-dimensional Mart{\'{\i}}nez Alonso--Shabat
  equation.}

\author{I.S. Krasil'shchik} \address{Trapeznikov Institute of
  Control Sciences
  \\
  65 Profsoyuznaya Street, Moscow 117997, Russia} \email{josephkra@gmail.com}

\author{O.I. Morozov} \address{Trapeznikov Institute of Control Sciences
  \\
  65 Profsoyuznaya Street, Moscow 117997, Russia} \email{oimorozov@gmail.com}


\date{\today
}



\begin{abstract}
  We study a Lagrangian extension of the 5d Mart{\'{\i}}nez Alonso--Shabat
  equation~$\mathcal{E}$
  \begin{equation*}
     u_{yz}=u_{tx}+u_y\,u_{xs}-u_x\,u_{ys}
   \end{equation*}
   that coincides with the cotangent equation~$\mathcal{T^*E}$ to the
   latter. We describe the Lie algebra structure of its symmetries (which
   happens to be quite nontrivial and is described in terms of deformations)
   and construct two families of recursion operators for
   symmetries. Each family depends on two parameters. 
We prove that all the operators from the first family are hereditary, 
but not    compatible in the sense of the Nijenhuis bracket. We also construct two new
   parametric Lax pairs that depend on higher-order derivatives of the
   unknown functions.
\end{abstract}
\maketitle
\tableofcontents
\newpage

\section{Introduction}\label{sec:introduction}
In what follows, we consider the two-component system
\begin{equation}
  \label{eq:1}
  \begin{array}{rcl}
    u_{yz}&=&u_{tx}+u_y\,u_{xs}-u_x\,u_{ys},\\
    v_{yz} &=& v_{tx}+u_y\,v_{xs}-u_x\,v_{ys}+2\,u_{ys}\,v_x-2\,u_{xs}\,v_y,
  \end{array}
\end{equation}
which is nothing else but the cotangent equation (see,
e.g.,~\cite{KrasilshchikVerbovetskiyVitolo2017} and references therein) to the
$5$-dimensional  Mart{\'{\i}}nez Alonso--Shabat equation
\begin{equation}
  \label{eq:2}
  u_{yz} = u_{tx}+u_y\,u_{xs}-u_x\,u_{ys},
\end{equation}
see~\cite{MartinezAlonsoShabat2004, ManakovSantini2006, ManakovSantini2014,
BKMV2015}.
We describe the Lie algebra structure of its symmetries  and construct two families of 
recursion operators for  symmetries. Each family depends on two parameters. 
It is proved that all the operators from the first family are hereditary, but not    
compatible in the sense of the Nijenhuis bracket. We also construct two new
parametric Lax pairs that depend on higher-order derivatives of the
unknown functions.  

Equation \eqref{eq:2} admits a number of symmetry reductions. The
4-dimensional reductions include the reduced quasi-classical self-dual
Yang--Mills equation~\cite{FerapontovKhusnutdinova2004}
\begin{equation}
  u_{yz} = u_{tx}+u_y\,u_{xx}-u_x\,u_{xy},
\label{FKh4}
\end{equation}
the four-dimensional Mart{\'{\i}}nez Alonso--Shabat equation
\cite{MartinezAlonsoShabat2004}
\begin{equation}
  u_{ty} = u_z\,u_{xy} - u_y\,u_{xz},
\label{MASh4}
\end{equation}
and the 4D  universal hierarchy equation  \cite{BogdanovPavlov2017}   
\begin{equation}
  u_{zz} = u_{tx} + u_z u_{xy} - u_x u_{yz}.
\label{4D_UHE}
\end{equation}
In their turn, these equations admit 3D symmetry reductions: 
the hyper-CR equation for Ein\-stein--Weyl structures
 \cite{Kuzmina1967,Mikhalev1992,Pavlov2003,Dunajski2004}
\begin{equation}
  u_{yy} = u_{tx} + u_y\,u_{xx} - u_x\,u_{xy},
\label{Pavlov_eq}
\end{equation} 
the 3D rdDym equation \cite{Blaszak2002}
\begin{equation}
  u_{ty} = u_x\,u_{xy} - u_y \,u_{xx},
\label{Oe}
\end{equation}
the 3D universal hierarchy equation 
\cite{MartinezAlonsoShabat2002, MartinezAlonsoShabat2004}
\begin{equation}
  u_{yy} = u_y\,u_{tx} - u_x \,u_{ty},
\label{uhe}
\end{equation}
and the modified Veronese web equation
\cite{AdlerShabat2007,FerapontovMoss2015}
\begin{equation}
  u_{ty} = u_t\,u_{xy} - u_y \,u_{tx}.
\label{mVwe}
\end{equation}
The Lax representation \eqref{extended_MASh5D_shadow_generated_covering} as
well as the recursion operator \eqref{lambda_recursion_operator_for_MASh5D},
\eqref{lambda_recursion_operator_for_adjoint_linearized_MASh5D} survive in
these reductions when $\kappa=\mu=0$ and thus provide Lax representations and
recursion operators for the cotangent extensions of the above reduced
equations. The questions of whether these cotangent extensions possess
parametric families of Lax representations and recursion operators, and
whether the recursion operators are hereditary, are more subtle. The
affirmative answer to the second question for Eqs.~\eqref{Oe} and
\eqref{Pavlov_eq} was obtained in
\cite{KrasilshchikVerbovetskiy2022,Krasilshchik2022}.  The Lax representation
of Eq. \eqref{MASh4} with two non-removable parameters \cite{Morozov2014} was
used to construct new recursion operators for this equation in a recent
preprint~\cite{Vojcak-2022}.

The other reductions of Eq. \eqref{MASh5D} as well as some other
multi-dimensional integrable \PDE\ will be considered in the forthcoming
parts of the work.

In Section~\ref{sec:prel-notat}, we very briefly discuss the approach adopted
in the study. A more detailed information can be found in the
monographs~\cite{Bocharov1999, KrasilshchikVerbovetskiyVitolo2017} and
in~\cite{Krasilshchik2022}. Section~\ref{sec:5d-martinez-alonso} contains the
main results. Namely, Subsection~\ref{sec:symmetries} deals with a full
description of the Lie algebra of symmetries of Eq.~\eqref{eq:1}, which
consists not only of point symmetries (a typical situation for
multi-dimensional systems), but contains one higher (of order~$3$) one.  In
Subsection~\ref{sec:lax-representations-1}, we find a two-parameter Lax pair
for System~\eqref{eq:1} which is used to construct two families of recursion
operators. We show that the operators from the first family are hereditary,
but pair-wise incompatible in the sense of the Nijenhuis bracket. The action
of the operators from the first family on symmetries is described in
Subsection~\ref{sec:actions}. In Subsection~\ref{sec:higher-order-lax}, using
specific properties of the above-mentioned higher symmetry, we construct two
families of Lax pairs that depend on variables of order~$5$ and $4$. Finally,
in Section~\ref{sec:concluding-remarks}, we discuss some the perspectives.

\section{Preliminaries and notation}\label{sec:prel-notat}

Our goal is to study various invariants (symmetry algebras, Lax pairs,
recursion operators, etc.) of canonical Lagrangian extensions associated to
multi-dimensional systems. The approach adopted here is based on the
geometrical theory of \PDE s (see~\cite{Bocharov1999} for the main definitions
and standard notation and~\cite{KrasilshchikVinogradov1989} for the details of
the nonlocal theory).

Namely, any \PDE\ (its infinite prolongation, to be more precise) is treated
as a submanifold $\mathcal{E}\subset J^\infty(\pi)$, where~$J^\infty(\pi)$ is
the space of infinite jets of some locally trivial vector bundle
$\pi\colon E\to M$, $\dim M = n$, $\rank\pi = m$. There exists a natural
projection $\pi_\infty\colon \mathcal{E}\to M$ and~$\mathcal{E}$ is always
endowed with an integrable $n$-dimensional $\pi_\infty$-horizontal
distribution $\mathcal{C}\subset T\mathcal{E}$ (the Cartan
distribution)\footnote{Since $\mathcal{C}$ is horizontal and $n$-dimensional,
  a flat connection in~$\pi_\infty$ is associated to this distribution (the
  Cartan connection).}. Locally,~$\mathcal{C}$ annihilates all the Cartan
forms $\omega_\sigma^j = du_\sigma^j - \sum_i u_{\sigma i}^j\,dx^i$.

A $\pi_\infty$-vertical vector field~$Z$ on~$\mathcal{E}$ is a symmetry if
$[Z, \mathcal{C}]\subset \mathcal{C}$. Any such a~$Z$ has the form of an
evolutionary derivation~$\Ev_\phi$, where~$\phi$ is the generating section
of~$Z$; we do not distinguish between~$Z$ and~$\phi$. The space of
symmetries~$\sym\mathcal{E}$ carries the structure of a Lie algebra with
respect to the commutator. The corresponding bracket on generating sections is
denoted by~$\{\cdot\,,\cdot\}$ and is called the Jacobi
bracket. Let~$\mathcal{E}$ be given by the system
$\{F = 0,\ F = (F^1, \dots, F^r)\}$. To find symmetries, one needs to solve
the linear system
\begin{equation*}
  \ell_{\mathcal{E}}(\phi) = 0,
\end{equation*}
where~$\ell_{\mathcal{E}}$ is the restriction of the linearization
operator~$\ell_F$ to~$\mathcal{E}$.

Consider an overdetermined \PDE~$\mathcal{W}$ whose compatibility conditions
coincide with~$\mathcal{E}$. Then the surjection $\tau\colon \mathcal{W} \to
\mathcal{E}$ is a covering. Coordinates in the fiber of~$\tau$ are called
nonlocal variables. If the defining equations of~$\mathcal{W}$ are linear in
nonlocal variable, this covering is a Lax pair for~$\mathcal{E}$. An
$\mathbb{R}$-linear derivation $S\colon C^\infty(\mathcal{E}) \to
C^\infty(\mathcal{W})$ is a nonlocal $\tau$-shadow if it preserves Cartan
distributions (commutes with the action of Cartan connections). The defining
equation for a shadow~$S = \tilde{\Ev}_\phi$ is
\begin{equation*}
  \tilde{\ell}_{\mathcal{E}}(\phi) = 0,
\end{equation*}
where $\tilde{\ell}_{\mathcal{E}}$ is the lift of~$\ell_{\mathcal{E}}$
from~$\mathcal{E}$ to~$\mathcal{W}$, while~$\phi$ is the $m$-component
generating section of~$S$ that lives on~$\mathcal{W}$.

Let $\tau_\lambda\colon \mathcal{W}_\lambda \to \mathcal{E}$ be a
$\lambda$-parametrized family of coverings, $\lambda\in\mathbb{R}$. The
parameter~$\lambda$ is said to be non-removable if the
coverings~$\tau_\lambda$ are pair-wise non-equivalent. There exists a regular
way to insert a parameter in an arbitrary covering,
see~\cite{KrasilshchikVinogradov1989}. Namely, let~$Z$ be an integrable
symmetry of~$\mathcal{E}$, i.e., such that it possesses local
trajectories. Assume that~$Z$ cannot be lifted to~$\mathcal{W}$, i.e., there
exists no symmetry~$\tilde{Z}$ of~$\mathcal{W}$ such that
$\tau_*(\tilde{Z}) = Z$. Then $\exp(\lambda Z)$ generates the desired family
of coverings.

The following construction underlies our approach to recursion operators for
symmetries of~$\mathcal{E}$. Introduce the bundle
$\mathbf{t}\colon \mathcal{TE} = T\mathcal{E}/\mathcal{C} \to M$ which is
called the tangent covering of~$\mathcal{E}$. In coordinates, if~$\mathcal{E}$
is given by $\{F = 0\}$, then~$\mathcal{TE}$ is defined by the system
\begin{equation*}
  F(u) = 0, \qquad \ell_{\mathcal{E}}(q) = 0,
\end{equation*}
where $q = (q^1,\dots,q^m)$ is a new unknown which is assumed to be odd (of
parity~$1$). The algebra of super-functions on~$\mathcal{TE}$ is identified
with the Grassmann algebra~$\Lambda_v^*(\mathcal{E})$ of Cartan forms
on~$\mathcal{E}$. The Cartan differential
\begin{equation*}
  d_{\mathcal{C}}\colon \Lambda_v^i(\mathcal{E})\to
  \Lambda_v^{i+1}(\mathcal{E}),\qquad d_{\mathcal{C}}(f) = \sum\frac{\partial
    f}{\partial u_\sigma^j}\omega_\sigma^j,\quad
  d_{\mathcal{C}}(\omega_\sigma^j) = 0,
\end{equation*}
defines a canonical nilpotent vector field~$\dC$ on~$\mathcal{TE}$ of
parity~$1$. Sections of~$\mathbf{t}$ that preserve the Cartan distributions
coincide with symmetries of~$\mathcal{E}$.

Let~$\tau\colon \mathcal{W}\to \mathcal{TE}$ be a covering with fibers of
parity~$1$. Let also~$S$ be a $\mathbf{t}\circ\tau$-shadow linear in the
nonlocal variables of~$\tau$. Then it defines another covering $\tau_S\colon
\mathcal{W} \to \mathcal{TE}$ and we obtain the diagram
\begin{equation}\label{eq:3}
  \mathcal{R}\colon\qquad
  \xymatrix{\mathcal{TE}&\mathcal{W}
    \ar_-{\tau_s}[l]\ar[r]^-\tau&\mathcal{TE}\rlap{,} 
  }
\end{equation}
i.e., a recursion operator (a B\"{a}cklund auto-transformation
of~$\mathcal{TE}$), cf.~\cite{Mar-another}. We say that~$\mathcal{R}$ is a
regular operator if there exists a symmetry
$\tilde{S} = \Ev_\phi\in \sym\mathcal{W}$ that extends~$S$, i.e., such that
the restriction of~$\tilde{S}$ to~$C^\infty(\mathcal{E})$ coincides with~$S$.

Consider    two   regular    operators    of    the   form~\eqref{eq:3}    and
let~$\Ev_{\phi_1}$, $\Ev_{\phi_2}$  be the corresponding symmetries.  Then the
generating section~$\ldb\phi_1,\phi_2\rdb$ of their super-commutator
\begin{equation*}
  [\Ev_{\phi_1}, \Ev_{\phi_2}] = \Ev_{\phi_1}\circ\Ev_{\phi_2} +
  \Ev_{\phi_2}\circ\Ev_{\phi_1} 
\end{equation*}
(a symmetry of parity~$2$) is called their Nijenhuis bracket. An operator is
hereditary if~$\ldb\phi, \phi\rdb = 0$, two operators are compatible
if~$\ldb\phi_1, \phi_2\rdb = 0$.

Regularity is established in two steps: lift to~$\mathcal{TE}$ and subsequent
lift to~$\mathcal{W}$. First, using the interpretation of functions
on~$\mathcal{TE}$ as Cartan forms, we set
\begin{equation*}
  S(\omega_\sigma^j) = S(\dC(u_\sigma^j)) = -\dC(S(u_\sigma^j)).
\end{equation*}
The action of~$\dC$ on local variables is known. To compute its values at
nonlocal ones, we take the defining equations of the covering~$\tau$
\begin{equation*}
  w_{x^i}^j = W_i^j,\qquad i= 1,\dots,n.\quad j = 1,\dots,\rank\tau,
\end{equation*}
and apply~$\dC$ to them:
\begin{equation}
  \label{eq:4}
  D_{x^i}(\dC(w^j)) = \dC(W_i^j),
\end{equation}
where~$D_{x^i}$ are the total derivatives. Solving Equation~\eqref{eq:4} with
respect to~$\dC(w^j)$ (if possible) gives the desired result. In a similar
way, to find the action of~$S$ on~$w^j$, we solve the system
\begin{equation*}
  D_{x^i}(S(w^j)) = S(W_i^j)
\end{equation*}
with respect to~$S(w^j)$.

To finish this introductory part, let us say a few words about the extensions
we work with below. Let~$\mathcal{E}$ be an equation given by $\{F = 0\}$ in
local coordinates. Then the system
\begin{equation}
  \label{eq:5}
  \mathcal{T}^*\mathcal{E}\colon\qquad\ell_{\mathcal{E}}^*(p) = 0,\quad F(u) = 0,
\end{equation}
where~$\ell_{\mathcal{E}}^*$ is the adjoint operator, is called the cotangent
equation to~$\mathcal{E}$. System~\eqref{eq:5} is always a Lagrangian one with
the Lagrangian density~$\mathcal{L} = \langle p,
F\rangle\,dx^1\wedge\dots\wedge \,dx^n$. Though $\mathcal{T}^*\mathcal{E}$ is
defined in coordinates, it can be shown that when~$\mathcal{E}$ is a two-line
equation (see~\cite{Vin-C-spec}), the object is well defined,
see~\cite{KrasilshchikVerbovetskiyVitolo2017} for the proof.

\section{The 5D Mart{\'{\i}}nez Alonso--Shabat
  equation}\label{sec:5d-martinez-alonso} 

We consider the five-dimensional Mart{\'{\i}}nez Alonso--Shabat
equation~$\mathcal{E}$
\begin{equation}
  u_{yz}=u_{tx}+u_y\,u_{xs}-u_x\,u_{ys}.
\label{MASh5D}
\end{equation}
The linearization of this equation reads
\begin{equation}
  q_{yz}=q_{tx}+u_y\,q_{xs}+u_{xs}\,q_y-u_x\,q_{ys}-u_{ys}\,q_x.
\label{linearized_MASh5D}
\end{equation}
The cotangent covering
$\tilde{\mathcal{E}} = \EuScript{T^{\,*}E}\to \mathcal{E}$ (see,
e.g.,~\cite{KrasilshchikVerbovetskiyVitolo2017}) is obtained by appending the
adjoint linearization
\begin{equation}
  v_{yz} = v_{tx}+u_y\,v_{xs}-u_x\,v_{ys}+2\,u_{ys}\,v_x-2\,u_{xs}\,v_y
\label{adjoint_linearized_MASh5D}
\end{equation}
to \eqref{MASh5D}.

\subsection{Symmetries}\label{sec:symmetries}
Symmetries of the system at hand are described by the following

\begin{proposition} The Lie algebra $\mathfrak{s}$ of the local infinitesimal
  symmetries of order $\leq 3$ for System \eqref{MASh5D},
  \eqref{adjoint_linearized_MASh5D} is generated by the functions
  \begin{align*}
    \psi_1 &= (x\,u_x+y\,u_y-u, x\,v_x+y\,v_y),\\
    \psi_2 &= (-t\,u_t-x\,u_x-2\,z\,u_z-u, -t\,v_t-x\,v_x-2\,z\,v_z),\\
    \psi_3 &= (-t\,u_t+x\,u_x-y\,u_y+z\,u_z, -t\,v_t+x\,v_x-y\,v_y+z\,v_z),\\
    \psi_4 &= (-t\,z\,u_t-x\,z\,u_x-t\,x\,u_y-z^2\,u_z-z\,u,\\
           &\ \ -z\,v-x\,z\,v_x-t\,z\,v_t-t\,x\,v_y-z^2\,v_z),\\
    \psi_5 &= (-y\,u_x-t\,u_z, -y\,v_x-t\,v_z),\\
    \psi_6 &= (t^2\,u_t+y\,z\,u_x+t\,y\,u_y+t\,z\,u_z+t\,u,\\
           &\ \ t^2\,v_t +y\,z\,v_x+t\,y\,v_y+t\,z\,v_z+t\,v),\\
    \psi_7 &= (u_z, v_z),\\
    \psi_8 &= (-z\,u_t-x\,u_y, -z\,v_t-x\,v_y),\\
    \psi_9 &= (-u_t, -v_t),\\
    \psi_{10} &= (-z\,u_x-t\,u_y, -z\,v_x-t\,v_y),\\
    \psi_{11} &= (u_x, v_x),\\
    \psi_{12} &= (-u_y, -v_y),\\
    \psi_{13} &= (0, -u_{sss}),\\
    \psi_{14} &=(0,v),\\
    \varphi_{0,0}(A)&= (-Au_s+A_{s}\,u+A_{z}x+A_{t}y,
                      -A\,v_s-2\,A_{s}\,v),\\
    \varphi_{0,1}(A) &= (A,0),\\
    \varphi_{1,0}(A) &= (0, 2\,A\,u_s+A_{s}\,u+A_{z}\,x+A_{t}\,y),\\
    \varphi_{1,1}(A) &= (0,A),
  \end{align*} 
  where $A=A(t,z,s)$ and $B=B(t,z,s)$ below are arbitrary smooth functions of
  their arguments.  Besides\textup{,} the map $(t,x,y,z) \mapsto (z,y,x,t)$ is
  a discrete \textup{(}finite\textup{)} symmetry of this system.  There hold
  \begin{equation*}
    \{\varphi_{i,m}(A),\varphi_{j,n}(B)\}=
    \begin{cases}
      \varphi_{i+j,m+n}(AB_s-(1-3j)BA_s), & i+j \le 1, \\
      &m+n \le 1,  \\
      0, & \mathrm{otherwise}
    \end{cases}
  \end{equation*}
  for $i \le j$ and
  \begin{align*}
    \{\psi_1, \varphi_{i,m}(A)\}
    &= (i-m)\,\varphi_{i,m}(A),\\
    \{\psi_2, \varphi_{i,m}(A)\}
    &= \varphi_{i,m}(t\,A_t+2\,z\,A_z+(i-m)\,A),\\
    \{\psi_3, \varphi_{i,m}(A)\}
    &= \varphi_{i,m}(t\,A_t-z\,A_z),\\
    \{\psi_4, \varphi_{i,m}(A)\}
    &= \varphi_{i,m}(t^2\,A_t+z^2\,A_z+m\,z\,A),\\
    \{\psi_5, \varphi_{i,m}(A)\}
    &= \varphi_{i,m}(t\,A_z),\\
    \{\psi_6, \varphi_{i,m}(A)\}
    &= \varphi_{i,m}(-t^2\,A_t-t\,z\,A_z-m\,A),\\
    \{\psi_7, \varphi_{i,m}(A)\}
    &= \varphi_{i,m}(-z\,A_z),\\
    \{\psi_8, \varphi_{i,m}(A)\}
    &= \varphi_{i,m}(z\,A_t),\\
    \{\psi_9, \varphi_{i,m}(A)\}
    &= \varphi_{i,m}(A_t),\\
    \{\psi_{10},\varphi_{i,m}(A)\}
    &= \begin{cases}
      \varphi_{i,1}(t\,A_t+z\,A_z), & m=0,\\
      0, & \mathrm{otherwise},
    \end{cases}\\
    \{\psi_{11},\varphi_{i,m}(A)\}
    &=
      \begin{cases}
        \varphi_{i,1}(-z\,A_z), & m=0,\\
        0, & \mathrm{otherwise},
      \end{cases}\\
    \{\psi_{12},\varphi_{i,m}(A)\}
    &=
      \begin{cases}
        \varphi_{i,1}(A_t), & m=0,\\
        0, & \mathrm{otherwise},
      \end{cases}\\
    \{\psi_{13},\varphi_{i,m}(A)\}
    &=
      \begin{cases}
        \varphi_{1,m}(A_{sss}), & i=0,\\
        0, & \mathrm{otherwise},
      \end{cases}\\
    \{\psi_{14},\varphi_{i,m}(A)\}
    &=-i\,\varphi_{i,m}(A),\\
    \{\psi_i,\psi_{14}\}
    &=0, \qquad i \in \{1, ... , 12\},\\
    \{\psi_i,\psi_{13}\}
    &=
      \begin{cases}
        -\psi_{13}, & i \in \{1, 2, 14\},\\
        0, & \mathrm{otherwise}.
      \end{cases}
  \end{align*}
  The subalgebra
  $\langle \psi_1, ... , \psi_{12} \rangle \subset \mathfrak{s}$ is isomorphic
  to the Lie algebra $\mathfrak{gl}(3,\mathbb{R}) \ltimes \mathbb{R}^3$ with
  the isomorphism given by the map
  \begin{equation*}
    \fl
    \begin{array}{lcl}
      (\psi_1, ... , \psi_{12}) 
      &\mapsto  &(E_{11}+E_{22}+E_{33},
                  E_{11}-E_{22},E_{22}-E_{33},E_{12},E_{23},\\
      & & E_{13},E_{21},E_{32},E_{31},e_1,e_2,e_3),
    \end{array}
  \end{equation*}
where $E_{ij} \in \mathfrak{gl}(3,\mathbb{R})$ is the matrix with the only
non-zero $(i,j)$-entry $1$ and for the vectors $e_k \in \mathbb{R}^3$ there
hold $E_{ij}\,e_k = \delta_{jk}\,e_i$.  \hfill $\Box$
\end{proposition}

The Lie algebra $\mathfrak{s}_0$ of the contact symmetries of
Eq.~\eqref{MASh5D} is generated by the first components of the symmetries
$\psi_1,\dots,\psi_{12}$, $\varphi_{0,0}(A)$, and $\varphi_{0,1}(A)$.

Let
$\mathfrak{w}= \Der(C^\infty(\mathbb{R})) = \{f(s)\,\partial_s \,\,\vert\,\, f
\in C^\infty(\mathbb{R})\}$ be the Lie algebra of smooth vector fields on
$\mathbb{R}$ (or, equivalently, $\mathbb{R}$-linear derivations
of~$C^\infty(\mathbb{R})$). In other words, this Lie algebra is the vector
space $C^\infty(\mathbb{R})$ endowed with the bracket $[f,g]_0$ $=$
$f\,g_s-g\,f_s$.  For $n \in \mathbb{N}$, consider the commutative uni\-tal
algebra of the truncated polynomials in the (formal) variable $h$ of degree
$n$: \, $\mathbb{R}_n[h] = \mathbb{R}[h] / \langle h^{n+1} =0\rangle$.  Then
\begin{equation*}
  \mathfrak{s}_0 \cong 
  (\mathfrak{gl}(3,\mathbb{R}) \ltimes \mathbb{R}^3) 
  \ltimes (C^{\infty}(\mathbb{R}^2) \otimes \mathbb{R}_1[h]\otimes \mathfrak{w}).
\end{equation*}
We need the following construction to describe the structure of the Lie
algebra $\mathfrak{s}$. 

The Lie algebra $\mathfrak{q}_{n,0} = \mathbb{R}_n[\tau] \otimes \mathfrak{w}$
admits the de\-for\-ma\-ti\-on\footnote[1]{For a full description of the
  deformations of the subalgebra
  $\mathbb{R}_n[\tau] \otimes \Der(\mathbb{R}[s]) \subset \mathfrak{q}_{n,0}$
  see \cite{Zusmanovich2003}.}  generated by the cocycle
$\Psi \in H^2(\mathfrak{q}_{n,0}, \mathfrak{q}_{n,0})$,
\begin{equation*}
  \Psi (\tau^k \otimes f, \tau^m \otimes g) =
  \begin{cases}
    \displaystyle{\tau^{k+m} \otimes (k\,f\,g_s-m\,g\,f_s)}, &  k+m \le n,\\
    0, & \mathrm{otherwise}.
  \end{cases}
\end{equation*}
For each $\varepsilon \neq 0$, this cocycle defines a new bracket
$[\cdot, \cdot]_{\varepsilon} = [\cdot, \cdot] +\varepsilon\,\Psi(\cdot,
\cdot)$ on the linear space
$\langle \tau^m \otimes f \,\vert \, m \le n, f \in C^\infty(\mathbb{R})
\rangle$. We denote the resulting Lie algebra by
$\mathfrak{q}_{n,\varepsilon}$.  In other words,
$\mathfrak{q}_{n,\varepsilon}$ is isomorphic to the linear space of functions
$F(s,\tau) = f_0(s)+\tau\,f_1(s)+\dots+\tau^n\,f_n(s)$,
$f_k \in C^\infty(\mathbb{R})$, equipped with the bracket
\begin{equation*}
  [F,G]_\varepsilon =F\,G_s-G\,F_s+\varepsilon\,\tau\,(F_\tau\,G_s-G_\tau\,F_s)
\end{equation*}
such that for \, $k > n$ \,there holds\, $\tau^k =0$ . Arguments similar to
\cite{DavidKamranLeviWinternitz1985} show that the subalgebra of
$\mathfrak{q}_{n,\varepsilon}$ obtained by replacing
$\Der(C^\infty(\mathbb{R}))$ to $\Der(\mathbb{R}[s,s^{-1}])$ in the definition
of $\mathfrak{w}$ is a proper subalgebra of the affine Kac--Moody Lie algebra
$\mathfrak{g}(A_m^{(1)})$ with the generalized Cartan matrix $A_m^{(1)}$
\cite{Kac1990} for some $m \ge n$.  Therefore the algebras
$\mathfrak{q}_{n,\varepsilon}$ are referred to as the Lie algebras of
Kac--Moody type.

The map
$D_0=\tau\,\partial_\tau \colon \mathfrak{q}_{n,\varepsilon} \rightarrow
\mathfrak{q}_{n,\varepsilon}$,
$D_0 \colon \tau^k \otimes f \mapsto k \,\tau^k \otimes f$, is an outer
derivation of the Lie algebra $\mathfrak{q}_{n,\varepsilon}$ for all
$n \in \mathbb{N}$ and $\varepsilon \in \mathbb{R}$. For some special values
of the parameters $n$ and $\varepsilon$ this algebra has other outer
derivations.

\begin{proposition}
  \label{sec:symmetries-1}
  The Lie algebra $\mathfrak{q}_{1,-3}$ admits the outer derivation
  \begin{equation*}
    D_1(\tau^k\otimes f) =
    \begin{cases}
      \tau\otimes f_{sss}, &k = 0,\\
      0, & k=1.
    \end{cases}
  \end{equation*}
\end{proposition}
\begin{proof}
  For $f, g \in C^\infty(\mathbb{R})$ and
  $D_1\colon \mathfrak{q}_{1,\varepsilon} \rightarrow
  \mathfrak{q}_{1,\varepsilon}$, there hold
  \begin{align*}
     &D_1([f,g]_{\varepsilon})-[D_1(f),g]_{\varepsilon}-[f,D_1(g)]_{\varepsilon}=
       (\varepsilon+3)\,\tau\,(f_{s}\,g_{sss}-g_{s}\,f_{sss}),\\
     &D_1([f,\tau\,g]_{\varepsilon})-[D_1(f),\tau\,g]_{\varepsilon}
       -[f,D_1(\tau\,g)]_{\varepsilon}
       \\
&=\tau^2\,((\varepsilon+3)\,f_s\,g_{sss}-(4\,\varepsilon+3)\,g_s\,f_{sss} -3\,\varepsilon\,f_{ss}\,g_{ss})=0, 
\\
    \intertext{and}
    &D_1([\tau\,f, \tau\,g]_{\varepsilon})-[D_1(\tau\,f),
    \tau\,g]_{\varepsilon}-[\tau\,f,
    D_1(\tau\,g)]_{\varepsilon}
\\
&=(4\,\varepsilon+3)\,\tau^3\,(f_{s}\,g_{sss}-g_{s}\,f_{sss})=0,
  \end{align*}
  therefore $D_1 \in \Der(\mathfrak{q}_{1,\varepsilon})$ if and only if
  $\varepsilon =-3$.  Furthermore,
  \begin{equation*}
    D_1(f) - [g_1+\tau\,g_2,f]_{\varepsilon} = f\,g_{1,s} -g_1\,f_{s}
    +\tau\,(f_{sss} - (1+\varepsilon)\,g_2\,f_{s} +f\,g_{2,s})
  \end{equation*}
  for all $f, g_1, g_2 \in C^\infty(\mathbb{R})$, and there is no choice of
  functions $g_1$, $g_2$ that would eliminate the term $\tau\,f_{sss}$ in the
  right-hand side of the last equation. Thus
  $D_1 \in \mathrm{Der}_{\mathrm{out}}(\mathfrak{q}_{1,-3})$.
\end{proof}

Denote by $\mathfrak{Q}$ the two-dimensional `right' extension, \cite[\S
1.4.4]{Fuks1984}, $\langle D_0, D_1\rangle \ltimes \mathfrak{q}_{1,-3}$ of the
Lie algebra $\mathfrak{q}_{1,-3}$ associated to the derivations $D_0$ and
$D_1$. As a vector space
$\mathfrak{Q} = \langle w_0, w_1 \rangle \oplus \mathfrak{q}_{1,-3}$, and the
bracket on $\mathfrak{q}_{1,-3}$ is extended to the new basis elements $w_0$
and $w_1$ by the formulas $[w_0,f+\tau\,g]_{-3}$ $=$ $D_0(f+\tau\,g)$ $=$
$\tau\,g$, $[w_1,f+\tau\,g]_{-3}$ $=$ $D_1(f+\tau\,g)$ $=$ $\tau\,f_{sss}$,
and $[w_0,w_1]_{-3}=w_1$. Then we have
\begin{equation*}
  \mathfrak{s} \cong 
  (\mathfrak{gl}(3,\mathbb{R}) \ltimes \mathbb{R}^3)  
  \ltimes 
  (C^{\infty}(\mathbb{R}^2) \otimes \mathbb{R}_1[h]\otimes \mathfrak{Q}),
\end{equation*}
where the derivatives $D_0$ and $D_1$ correspond to the symmetries $\psi_{14}$
and $\psi_{13}$.

\subsection{Lax representations}\label{sec:lax-representations-1}

Eq.~\eqref{MASh5D} admits the Lax representation
\begin{equation}\label{MASh5D_covering_Lambda}
  \begin{array}{lcl}
    w_t &=& \Lambda\,w_y-u_y\,w_s,
    \\
    w_z &=& \Lambda\,w_x-u_x\,w_s,
  \end{array}
\end{equation}
where
\begin{equation*}
  \Lambda= \frac{\lambda-\kappa\,y-\mu\,x }{1+\kappa\,t+\mu\,z}
\end{equation*}
with the parameters $\kappa$, $\mu$, $\lambda \in \mathbb{R}$
(see~\cite{BKMV2015}). When $\kappa=\mu=0$, this Lax representation coincides
with the Lax representation
\begin{equation}\label{MASh5D_main_covering}
  \begin{array}{lcl}
    w_t &=& \lambda\,w_y-u_y\,w_s,
    \\
    w_z &=& \lambda\,w_x-u_x\,w_s,
  \end{array}
\end{equation} 
which was found in \cite{ManakovSantini2006} and used intensively in
\cite{ManakovSantini2014}. The parameters $\kappa$, $\mu$, and $\lambda$ are
non-removable, that is, the differential coverings defined by system
\eqref{MASh5D_covering_Lambda} with the different constant values of $\kappa$,
$\mu$, and $\lambda$ are not equivalent. To ensure this, we note that the
first components of the symmetries $\psi_4$, $\psi_6$, and $\psi_{10}$ from
Proposition 1 do not admit lifts to symmetries of the system
\begin{equation}\label{MASh5D_special_covering}
  \begin{array}{lcl}
    w_t &=& -u_y\,w_s,
    \\
    w_z &=& -u_x\,w_s.
  \end{array}
\end{equation}
Then we take the vector field associated with the first component of the
linear combination $-\mu\,\psi_4+\kappa\,\psi_6+\lambda\,\psi_{10}$. The flow
of its first prolongation maps system \eqref{MASh5D_special_covering} to
system \eqref{MASh5D_covering_Lambda}.  In accordance with \cite[\S\S~3.2,
3.6]{KrasilshchikVinogradov1989},
\cite{Krasilshchik2000,IgoninKrasilshchik2000,Marvan2002,IgoninKerstenKrasilshchik2002}
this proves the claim.

In what follows, we will use another Lax representation for
Eq. \eqref{MASh5D}.  We observe that the function
\begin{equation}
  \theta=\frac{1}{(\lambda-\kappa\,y-\mu\,x)\,w_s}
\label{theta_shadow}
\end{equation}
is a shadow of a nonlocal symmetry for Eq. \eqref{MASh5D} in the covering
\eqref{MASh5D_covering_Lambda}.  We express $w_s$ from~\eqref{theta_shadow},
differentiate the result with respect to $t$ and $z$, and substitute
to~Eqs.~\eqref{MASh5D_covering_Lambda}.  This yields a new covering
\begin{equation}\label{MASh5D_shadow_generated_covering}
  \begin{array}{lcl}
    \theta_t &=& \displaystyle{\Lambda\,\theta_y-u_y\,\theta_s
                 + \left(u_{ys}-\frac{\kappa}{\kappa\,t+\mu\,z+1}\right)
                 \,\theta},  
    \\[14pt]
    \theta_z &=& \displaystyle{\Lambda\,\theta_x-u_x\,\theta_s
                 +\left(u_{xs}-\frac{\mu}{\kappa\,t+\mu\,z+1}\right)\,\theta}
  \end{array}
\end{equation}
over Eq.~\eqref{MASh5D}. This covering admits an extension to a covering over 
$\EuScript{T}^{\,}\EuScript{E}$:

\begin{proposition} Systems \eqref{MASh5D_shadow_generated_covering} and
  \begin{equation}\label{extended_MASh5D_shadow_generated_covering}
    \begin{array}{lcl}
      \omega_t
      &=& \displaystyle{\Lambda\,\omega_y-u_y\,\omega_s
          -\left(2\,u_{ys}+\frac{\kappa}{\kappa\,t+\mu\,z+1}\right)\,\omega}
          +2\,v_y\,\theta_s+v_{ys}\,\theta,
      \\[14pt]
      \omega_z
      &=& \displaystyle{\Lambda\,\omega_x-u_x\,\omega_s
          -\left(2\,u_{xs}+\frac{\mu}{\kappa\,t+\mu\,z+1}\right)\,\omega}
          +2\,v_x\,\theta_s+v_{xs}\,\theta
    \end{array}
  \end{equation}
  define a Lax representation for System~\eqref{MASh5D}\textup{,}
  \eqref{adjoint_linearized_MASh5D}.  \hfill $\Box$
\end{proposition}
Moreover,  we have
\begin{proposition}\label{sec:lax-representations}
  A solution $(\theta, \omega)$ of System
  \eqref{MASh5D_shadow_generated_covering}\textup{,}
  \eqref{extended_MASh5D_shadow_generated_covering} is a shadow of a
  non\-lo\-cal symmetry for System~\eqref{MASh5D}\textup{,}
  \eqref{adjoint_linearized_MASh5D}.  \hfill $\Box$
\end{proposition}

Another  covering  over 
$\EuScript{T}^{*}\EuScript{E}$ is defined by a lift of System  \eqref{MASh5D_main_covering}:

\begin{proposition}\label{sec:lax-representations}
  Systems \eqref{MASh5D_main_covering} and 
\begin{equation}\label{MASh5D_lift_of_main_covering}
  \begin{array}{lcl}
    W_t &=& \lambda\,W_y-u_y\,W_s -3 \,u_{ys} \,W+v_y\,w_s,
    \\
    W_z &=& \lambda\,W_x-u_x\,W_s-3\,u_{xs}\,W+v_x\,w_s.
  \end{array}
\end{equation}
provide a Lax representation for System \eqref{MASh5D}\textup{,}
\eqref{adjoint_linearized_MASh5D}.  \hfill $\Box$
\end{proposition}

In Section \ref{sec:actions} below we will use nonlocal variables of the
so-called negative and positive coverings, see~\cite{BKMV2015}, associated to
covering \eqref{MASh5D_main_covering},
\eqref{MASh5D_lift_of_main_covering}. To construct the negative covering, we
substitute the formal expansions $w=\sum_{n \ge 0} \lambda^{-n}\,w_{-n}$,
$W=\sum_{n \ge 0} \lambda^{-n}\,W_{-n}$ into \eqref{MASh5D_main_covering},
\eqref{MASh5D_lift_of_main_covering} and obtain the infinite tower of Abelian
two-component coverings (nonlocal conservation laws) given by equations
\begin{equation*}
  \begin{array}{lcl}
    w_{0,x}&=& 0, 
    \\
    w_{0,y} &=&0,
    \\
    w_{-n-1,x} &=& w_{-n,z}+u_x\,w_{-n,s},
    \\
    w_{-n-1,y} &=& w_{-n,t}+u_y\,w_{-n,s},
  \end{array}
\end{equation*}
and 
\begin{equation*}
  \begin{array}{lcl}
    W_{0,x}&=& 0, 
    \\
    W_{0,y} &=&0,
    \\
    W_{-n-1,x} &=& W_{-n,z}+u_x\,W_{-n,s}+3\,u_{xs}\,W_{-n}-v_x\,w_{-n,s},
    \\
    W_{-n-1,y} &=& W_{-n,t}+u_y\,W_{-n,s}+3\,u_{ys}\,W_{-n}-v_y\,w_{-n,s}.
  \end{array}
\end{equation*}
In particular, if we put $w_0=s$, $w_{-1}=u$, $W_0=0$,
$W_{-1}=-v$ and denote $p=w_{-2}$,
$r=-W_{-2}+3\,u_s\,v$, we obtain the local conservation law
\begin{equation*}
  \begin{array}{lcl}
    p_{x} &=& u_{z}+u_x\,u_{s},
    \\
    p_{y} &=& u_{t}+u_y\,u_{s},
    \\
    r_{x} &=& v_{z}+u_x\,v_{s}-2\,u_s\,v_x,
    \\
    r_{y} &=& v_{t}+u_y\,v_{s}-2\,u_s\,v_y
  \end{array}
\end{equation*}
for $\EuScript{T}^{*}\EuScript{E}$.

Likewise, the positive covering is generated by the expansions
$w=\sum_{n \ge 0} \lambda^{n}\,w_{n}$, $W=\sum_{n \ge 0}
\lambda^{n}\,W_{n}$ that produce systems
\begin{equation*}
  \begin{array}{lcl}
    w_{0,t}&=& -u_y\,w_{0,s}, 
    \\
    w_{0,z} &=&-u_x\,w_{0,s},
    \\
    w_{n+1,t} &=& -u_y\,w_{n+1,s}+w_{n,y},
    \\
    w_{n+1,z} &=& -u_x\,w_{n+1,s}+w_{n,x}
  \end{array}
\end{equation*}
and 
\begin{equation*}
  \begin{array}{lcl}
    W_{0,t}&=& -u_y\,W_{0,s}+v_y\,w_{0,s}-3\,u_{ys}\,W_{0}, 
    \\
    W_{0,z}&=& -u_x\,W_{0,s}+v_x\,w_{0,s}-3\,u_{xs}\,W_{0},
    \\
    W_{n+1,t}&=& -u_y\,W_{n+1,s}+v_y\,w_{n+1,s}-3\,u_{ys}\,W_{n+1}+W_{n,y}, 
    \\
    W_{n+1,z}&=& -u_x\,W_{n+1,s}+v_x\,w_{n+1,s}-3\,u_{xs}\,W_{n+1}+W_{n,x}.
  \end{array}
\end{equation*}

\subsection{Recursion operators}\label{sec:recursion-operators-2}

In view of Proposition~\ref{sec:lax-representations}, we can use
Systems~\eqref{MASh5D_shadow_generated_covering},
\eqref{extended_MASh5D_shadow_generated_covering} for deriving recursion
operators for symmetries of Eqs. \eqref{MASh5D},
\eqref{adjoint_linearized_MASh5D} by the method of \cite{Sergyeyev2017}, see
also
\cite{KrasilshchikKersten1994,KrasilshchikKersten1995,MalykhNutkuSheftel2004,MarvanSergyeyev2012,MorozovSergyeyev2014,
  KruglikovMorozov2015}.  Since these equations are independent on the
parameters $\kappa$, $\mu$, $\lambda$, we expand $\theta$, $\omega$ in the
formal Laurent series with respect to one of the parameters, substitute this
series into systems \eqref{MASh5D_shadow_generated_covering},
\eqref{extended_MASh5D_shadow_generated_covering}, and collect terms at the
same powers of the parameter.  So, if we choose the parameter $\lambda$, we
substitute $\theta =\sum_{n \in \mathbb{Z}} \theta_n\,\lambda^n$,
$\omega =\sum_{n \in \mathbb{Z}}\omega_n \lambda^n$, then collect the terms at
$\lambda^n$ for a fixed arbitrary $n \in \mathbb{Z}$, and finally rename
$\theta_{n-1} \mapsto \tilde{\theta}$, $\theta_{n} \mapsto \theta$,
$\omega_{n-1} \mapsto \tilde{\omega}$, $\omega_{n} \mapsto \omega$. This gives
two systems
\begin{equation}\label{lambda_recursion_operator_for_MASh5D}
  \begin{array}{lcl}
    \tilde{\theta}_x
    &=& 
        -(\kappa\,t+\mu\,z+1)\,(\theta_z+u_x\,\theta_s-u_{xs}\,\theta)
        -(\kappa\,y+\mu\,x) \, \theta_x  -\mu\,\theta,
    \\
    \tilde{\theta}_y
    &=& -(\kappa\,t+\mu\,z+1)\,(\theta_t+u_y\,\theta_s-u_{ys}\,\theta)
        -(\kappa\,y+\mu\,x) \, \theta_y  -\kappa\,\theta
  \end{array}
\end{equation}  
and
\begin{equation}\label{lambda_recursion_operator_for_adjoint_linearized_MASh5D}
  \begin{array}{lcl}
    \tilde{\omega}_x
    &=&  
        -(\kappa\,t+\mu\,z+1)\,(\omega_z+u_x\,\omega_s+2\,u_{xs}\,
        \omega-2\,v_x\,\theta_s-v_{xs}\,\theta)   
    \\
    &&
       -(\kappa\,y+\mu\,x) \, \omega_x
       -\mu\,\omega,
    \\
    \tilde{\omega}_y
    &=&  
        -(\kappa\,t+\mu\,z+1)\,(\omega_t+u_y\,\omega_s+2\,u_{ys}\,
        \omega-2\,v_y\,\theta_s-v_{ys}\,\theta)  
    \\
    &&
       -(\kappa\,y+\mu\,x) \, \omega_y
       -\kappa\,\omega.
  \end{array}
\end{equation}

Likewise, expanding $\theta$ and $\omega$ into the Laurent series with respect
to $\mu$ gives the systems
\begin{equation}\label{beta_recursion_operator_for_MASh5D}
  \begin{array}{lcl}
    \tilde{\theta}_x
    &=&-\displaystyle{\frac{1}{x}}\,
        \left(z\,(\tilde{\theta}_z+u_x\,\tilde{\theta}_s-u_{xs}\,
        \tilde{\theta})+\tilde{\theta} 
        +(\kappa\,y-\lambda)\,\theta_x\right.
    \\
    &&
       \qquad\left.\phantom{\tilde{\theta}}
       +(\kappa\,t+1)\,(\theta_z+u_x\,\theta_s-u_{xs}\,\theta)\right),
    \\
    \tilde{\theta}_y
    &=&-\displaystyle{\frac{1}{x}}\,
        \left(z\,(\tilde{\theta}_t+u_y\,\tilde{\theta}_s-u_{ys}\,
        \tilde{\theta})+\kappa\,\theta 
        +(\kappa\,y-\lambda)\,\theta_y\right.
    \\
    &&
       \qquad\left.\phantom{\tilde{\theta}}
       +(\kappa\,t+1)\,(\theta_t+u_y\,\theta_s-u_{ys}\,\theta)\right),
  \end{array}
\end{equation}  
and
\begin{equation}\label{beta_recursion_operator_for_adjoint_linearized_MASh5D}
  \begin{array}{lcl}
    \tilde{\omega}_x
    &=&-\displaystyle{\frac{1}{x}}\,
        \left(z\,(\tilde{\omega}_z+u_x\,\tilde{\omega}_s
        +2\,u_{xs}\,\tilde{\omega}-2\,v_x\,\tilde{\theta}_s-v_{xs}\,
        \tilde{\theta})\right.
    \\
    &&
       +(\kappa\,t+1)\,(\omega_z+u_x\,\omega_s+2\,u_{xs}\,\omega-2\,v_x\,
       \theta_s-v_{xs}\,\theta) 
    \\
    &&
       \qquad\left.
       +\tilde{\theta}+(\kappa\,y-\lambda)\,\theta_x
       \right),
    \\ 
    \tilde{\omega}_y
    &=&-\displaystyle{\frac{1}{x}}\,
        \left(z\,(\tilde{\omega}_t+u_y\,\tilde{\omega}_s+2\,u_{ys}\,
        \tilde{\omega}-2\,v_y\,\tilde{\theta}_s-v_{ys}\,\tilde{\theta}) 
        \right.
    \\
    &&
       +(\kappa\,t+1)\,(\omega_t+u_y\,\omega_s+2\,u_{ys}\,\omega-2\,v_y\,
       \theta_s-v_{ys}\,\theta) 
    \\
    &&
       \qquad\left.
       +\kappa\,\tilde{\theta}+(\kappa\,y-\lambda)\,\theta_y
       \right).
  \end{array}
\end{equation}

The correspondence $\theta \mapsto \tilde{\theta}$ defined by
System~\eqref{lambda_recursion_operator_for_MASh5D} will be denoted by
$\tilde{\theta} = \mathcal{R}_{\kappa,\mu}(\theta)$, and its lift
$(\theta, \omega) \mapsto (\tilde{\theta}, \tilde{\omega})$ defined by
Systems~\eqref{lambda_recursion_operator_for_MASh5D},
\eqref{lambda_recursion_operator_for_adjoint_linearized_MASh5D} will be
denoted by
$(\tilde{\theta},\tilde{\omega}) =
\hat{\mathcal{R}}_{\kappa,\mu}(\theta,\omega)$.  Likewise, the correspondences
defines by System~\eqref{beta_recursion_operator_for_MASh5D} and by
Systems~\eqref{beta_recursion_operator_for_MASh5D},
\eqref{beta_recursion_operator_for_adjoint_linearized_MASh5D} will be denoted
by $\tilde{\theta}= \mathcal{S}_{\kappa,\lambda}(\theta)$ and
$(\tilde{\theta},\tilde{\omega}) =
\hat{\mathcal{S}}_{\kappa,\lambda}(\theta,\omega)$, respectively.  Since each
component of the above Lau\-rent series is a shadow of a symmetry, we obtain
the following assertion.

\begin{proposition}\label{sec:recursion-operators}
  Systems \eqref{lambda_recursion_operator_for_MASh5D} and
  \eqref{beta_recursion_operator_for_MASh5D} define two-parametric families of
  recursion operators $\EuScript{R}_{\kappa,\mu}$ and
  $\EuScript{S}_{\kappa,\lambda}$ for symmetries of
  Eq.~\eqref{MASh5D}. Likewise,
  Systems~\eqref{lambda_recursion_operator_for_MASh5D},
  \eqref{lambda_recursion_operator_for_adjoint_linearized_MASh5D} and
  \eqref{beta_recursion_operator_for_MASh5D},
  \eqref{beta_recursion_operator_for_adjoint_linearized_MASh5D} define
  two-parametric families of recursion operators
  $\hat{\EuScript{R}}_{\kappa,\mu}$ and $\hat{\EuScript{S}}_{\kappa,\lambda}$
  for symmetries of the cotangent extension of Eq.~\eqref{MASh5D}.  \hfill
  $\Box$
\end{proposition}

\begin{remark}\label{Remark1}
  From Systems~\eqref{lambda_recursion_operator_for_MASh5D} and
  \eqref{lambda_recursion_operator_for_adjoint_linearized_MASh5D} we have
  $\mathcal{R}_{\kappa,\mu} = (\kappa\,t+\mu\,z+1)\,\mathcal{R}_{0,0}
  +(\kappa\,y+\mu\,x)\,\mathbf{1}$ and
  $\hat{\mathcal{R}}_{\kappa,\mu} =
  (\kappa\,t+\mu\,z+1)\,\hat{\mathcal{R}}_{0,0}
  +(\kappa\,y+\mu\,x)\,\mathbf{1}$, where $\mathbf{1}$ is the identical map on
  the spaces of shadows of Eq.~\eqref{MASh5D} and its cotangent extension.
  \hfill $\diamond$
\end{remark}

\begin{remark}
  The recursion operators produced by expanding System
  \eqref{MASh5D_shadow_generated_covering},
  \eqref{extended_MASh5D_shadow_generated_covering} with respect to the
  parameter $\kappa$ are the compositions of the recursion operators
  $\EuScript{S}_{\mu,\lambda}$, $\hat{\EuScript{S}}_{\mu,\lambda}$,
  respectively, with the finite symmetry $(t,x,y,z) \mapsto (z,y,x,t)$ of
  Eq.~\eqref{MASh5D}.  \hfill $\diamond$
\end{remark}

Unlike Systems~\eqref{beta_recursion_operator_for_MASh5D} and
\eqref{beta_recursion_operator_for_adjoint_linearized_MASh5D}, systems
\eqref{lambda_recursion_operator_for_MASh5D} and
\eqref{lambda_recursion_operator_for_adjoint_linearized_MASh5D} define Abelian
coverings over Eqs. \eqref{MASh5D} and \eqref{adjoint_linearized_MASh5D}.
This allows one to use the results of
\cite{KrasilshchikVerbovetskiy2022,Krasilshchik2022} to prove the following

\begin{proposition}\label{sec:recursion-operators-1}
  The recursion operator $\EuScript{R}_{\kappa,\mu}$ is hereditary for each
  choice of the parameters $\kappa$\textup{,} $\mu$.  Two such operators with
  different values of the parameters $\kappa$, $\mu$ are compatible\textup{,}
  i.e.\textup{,} their Nijenhuis bracket vanishes.
  
  The recursion operator $\hat{\EuScript{R}}_{\kappa,\mu}$ is hereditary for
  each choice of the parameters $\kappa$\textup{,} $\mu$.  The operators
  $\hat{\EuScript{R}}_{\kappa_1,\mu_2}$ and
  $\hat{\EuScript{R}}_{\kappa_2,\mu_2}$ are compatible if and only if
  $\kappa_1 = \kappa_2$ and $\mu_1=\mu_2$.
\end{proposition}

\begin{proof}
  We introduce two functions $U = (\kappa\,t+\mu\,z+1)\,u_s\,\theta-\eta$ and
  $V = (\kappa\,t+\mu\,z+1)\,(v_s\,\theta-2\,u_s\,\omega)-\zeta$, where $\eta$
  and $\zeta$ are the nonlocal variables defined by the systems
  \begin{equation}    \label{eta_covering}
    \begin{array}{lcl}
      \eta_x &=& (\kappa\,t+\mu\,z+1)\,(\theta_z+u_s\,\theta_x+u_x\,\theta_s)
                 +(\kappa\,y+\mu\,x)\,\theta_x+\mu\,\theta,
      \\
      \eta_y &=& (\kappa\,t+\mu\,z+1)\,(\theta_t+u_s\,\theta_y+u_y\,\theta_s)
                 +(\kappa\,y+\mu\,x)\,\theta_y+\kappa\,\theta.
    \end{array}
  \end{equation}
  and
  \begin{equation}\label{zeta_covering}
    \begin{array}{lcl}
      \zeta_x
      &=& (\kappa\,t+\mu\,z+1)\,(\omega _z+
          u_x\,\omega_s-2\,u_s\,\omega_x+v_s\,\theta_x-2\,v_x\,\theta_s) 
      \\
      &&
         \qquad
         +(\kappa\,y+\mu\,x)\,\omega_x+\mu\,\omega,
      \\
      \zeta_y
      &=& (\kappa\,t+\mu\,z+1)\,(\omega_t+u_y\,\omega_s-2\,u_s\,\omega_y+
          v_s\,\theta_y-2\,v_y\,\theta_s) 
      \\
      &&
         \qquad
         +(\kappa\,y+\mu\,x)\,\omega_y+\kappa\,\omega.
    \end{array}    
  \end{equation}
  Here $\omega$, $\theta$, $\eta$, and $\zeta$ are understood as odd variables in the fibers of
  the tangent covering $\mathcal{T}\tilde{\mathcal{E}} \to
  \tilde{\mathcal{E}}$ and the symbol $\wedge$ below denotes their
  anti-commutative multiplication.
  
  Then $U$ is a shadow of a symmetry of Equation~\eqref{MASh5D}, while the
  pair $(U,V)$ is a shadow of a symmetry for System~\eqref{MASh5D},
  \eqref{adjoint_linearized_MASh5D}. By the technique used in the proof of
  Proposition 3 from~\cite{Krasilshchik2022} we obtain

  \begin{lemma}
    The shadow $U$ admits the lift $\Phi_{\kappa,\mu}$ to a symmetry of
    System~ \eqref{lambda_recursion_operator_for_MASh5D}.  The shadow $(U,V)$
    has the lift $\widehat{\Phi}_{\kappa,\mu}$ to a symmetry of
    System~\eqref{lambda_recursion_operator_for_MASh5D},
    \eqref{lambda_recursion_operator_for_adjoint_linearized_MASh5D}.  For the
    associated evolutionary vector fields $\mathbf{E}_{\Phi_{\kappa,\mu}}$ and
    $\mathbf{E}_{\widehat{\Phi}_{\kappa,\mu}}$ we have
    \begin{equation}\label{UV_lift}
      \begin{array}{lcl}
        \Theta
        &=&\mathbf{E}_{\Phi_{\kappa,\mu}}(\theta) =
            (\kappa\,t+\mu\,z+1)\,\theta \wedge \theta_s, 
        \\
        \Omega
        &=&\mathbf{E}_{\widehat{\Phi}_{\kappa,\mu}} (\omega) = 
            (\kappa\,t+\mu\,z+1)\,(\theta \wedge \omega_s+2\theta_s \wedge
            \omega), 
        \\
        H &=& \mathbf{E}_{\Phi_{\kappa,\mu}}(\eta) =
              (\kappa\,t+\mu\,z+1)\,\theta \wedge \eta_s, 
        \\
        Z
        &=& \mathbf{E}_{\widehat{\Phi}_{\kappa,\mu}}(\zeta)= 
            (\kappa\,t+\mu\,z+1)\,(\theta \wedge \zeta_s+2\,\eta_s \wedge \omega) 
        \\
        && 
           \qquad\qquad
           +6\,(\kappa\,t+\mu\,z+1)^2\,u_s\,\omega \wedge \theta_s.
      \end{array}
    \end{equation}
  \end{lemma}
  \begin{proof}[Proof of Lemma]
    Straightforward computations.
  \end{proof}
  To finish the prove of Proposition~\ref{sec:recursion-operators-1} we use
  formulas \eqref{UV_lift} to compute
  \begin{equation*}
    [\![\mathbf{E}_{\Phi_{\kappa_1,\mu_1}}, \mathbf{E}_{\Phi_{\kappa_2,\mu_2}}]\!] =0
  \end{equation*}
  and
  \begin{equation*}
    [\![\mathbf{E}_{\widehat{\Phi}_{\kappa_1,\mu_1}},
    \mathbf{E}_{\widehat{\Phi}_{\kappa_2,\mu_2}}]\!]   
    =\mathbf{E}_{\Psi}
  \end{equation*}
  with
  \begin{equation*}
    \Psi = 
    \left(
      \begin{array}{c}
        0
        \\
        -6\,((\kappa_1-\kappa_2)\,t+(\mu_1-\mu_2)\,z)^2\,u_s\,\omega\wedge
        \theta_s 
        \\
        0
        \\
        0
      \end{array}
    \right).
  \end{equation*}
The last equation shows that $[\![\mathbf{E}_{\widehat{\Phi}_{\kappa_1,\mu_1}},
    \mathbf{E}_{\widehat{\Phi}_{\kappa_2,\mu_2}}]\!] =0$ if and only if $\kappa_1=\kappa_2$ and 
$\mu_1=\mu_2$.
\end{proof}

\subsection{Actions}\label{sec:actions}

In this Subsection we study the actions of the recursion operators from
Proposition~\ref{sec:recursion-operators} in the simplest cases. Namely, we
consider the actions of the operators $\hat{\mathcal{R}}_{\kappa,\mu}$ and
$\mathcal{R}_{0,0}^{-1}$.  The results of
\cite{BKMV2016,BKMV2018,KrasilshchikMorozovVojcak2019,Vojcak-2022} show that
the adequate setting for studying the actions of the operators
$\mathcal{R}_{\kappa,\mu}^{-1}$, $\hat{\mathcal{R}}_{\kappa,\mu}^{-1}$,
$\mathcal{S}_{\kappa,\lambda}$, $\hat{\mathcal{S}}_{\kappa,\lambda}$,
$\mathcal{S}_{\kappa,\lambda}^{-1}$, and
$\hat{\mathcal{S}}_{\kappa,\lambda}^{-1}$ should include consideration of
nonlocal symmetries in various coverings.  We intent to deal with this issue
in the forthcoming research.

As usual, the image of zero under the action of a recursion operator is non-trivial. Below the actions of the recursion operators are computed modulo images of zero.

According to Remark \ref{Remark1} one can readily express the action of the operators $\hat{\mathcal{R}}_{\kappa,\mu}$
in terms of the action of $\hat{\mathcal{R}}_{0,0}$.  The direct computations show that the map defined the last recursion operator is given by the following formulas, where the nonlocal variables $p$ and $r$ were defined 
in Subsection \ref{sec:lax-representations-1}:
\begin{equation*}
  \fl
  \begin{array}{lcl}
    \psi_1 &\mapsto  & (2\,p-y\,u_t-u\,u_s-x\,u_z, r-y\,v_t-u\,v_s-x\,v_z), 
    \\
    \psi_2 &\mapsto    & (t\,p_t+2\,z\,p_z+2\,p-
                         u_s\,(t\,u_t+2\,z\,u_z+u)+x\,u_z,  
\\
           &&
              t\,r_t+2\,z\,r_z+r
              +2\,u_s\,(t\,v_t+2\,z\,v_z)
    \\ &&-v_s\,(t\,u_t+2\,z\,u_z+u)
          +x\,v_z
          ),
    \\
    \psi_3 &\mapsto   & ( t\,p_t-z\,p_z-x\,u_z+y\,u_t-u_s\,(t\,u_t-z\,u_z), 
    \\
           &&
              t\,r_t-z\,r_z-x\,v_z+y\,v_t
              +2\,u_s\,(t\,v_t-z\,v_z)-v_s\,(t\,u_t-z\,u_z)
              ), 
    \\
    \psi_4 &\mapsto & (z^2\,p_z+t\,z\,p_t+2\,z\,p
                      +(x-z\,u_s)\,(t\,u_t\,u_s+z\,u_z+u), 
    \\
           &&
              z^2\,r_z+t\,z\,r_t+2\,z\,r
              +(x+2\,z\,u_s)\,(t\,v_t+z\,v_z+v)
    \\
           &&
              -z\,v_s\,(t\,u_t+z\,u_z+u)), 
    \\
    \psi_5 &\mapsto & (t\,p_z-t\,u_z\,u_s+y\,u_z,
                      t\,r_z+y\,v_z-t\,u_z\,v_s+2\,t\,u_s\,v_z),  
    \\
    \psi_6 &\mapsto & (-t^2\,p_t-t\,z\,p_z-2\,t\,p
                      +(t\,u_s-y)\,(t\,u_t+z\,u_z+u),
    \\
           &&
              -t\,(t\,r_t+z\,r_z+2\,r)
              -(y+2\,t\,u_s)\,(t\,v_t+z\,v_z+v)
    \\
           &&
              +t\,v_s\,(t\,u_t+z\,u_z+u)),
    \\
    \psi_7 &\mapsto & (-p_z+u_z\,u_s, -r_z-2\,u_s\,v_z+u_z\,v_s), 
    \\
    \psi_8 &\mapsto & (z\,p_t-z\,u_t\,u_s+x\,u_t,
                      z\,r_t+2\,z\,u_s\,v_t+x\,v_t-z\,u_t\,v_s),  
    \\
    \psi_9 &\mapsto & (p_t-u_t\,u_s, r_t+2\,u_s\,v_t-u_t\,v_s), 
    \\
    \psi_{10} &\mapsto & ( t\,u_t+z\,u_z+u, t\,v_t+z\,v_z+v), 
    \\
    \psi_{11} &\mapsto & (-u_z, -v_z)=-\psi_7, 
    \\
    \psi_{12} &\mapsto & (u_t, v_t)=-\psi_9, 
    \\
    \psi_{13} &\mapsto & \left(0,
                         p_{sss}-u_s\,u_{sss}-\frac{3}{2}\,u_{ss}^2\right),  
    \\
    \psi_{14} &\mapsto & (0, -r-2\,u_s\,v),
    \\
    \phi_{0,0}(A) &\mapsto &
                             (A\,p_s-A_s\,p
                             +u_s\,(A_s\,u+A_t\,y+A_z\,x-A\,u_s)
    \\
           &&
              -\frac{1}{2}\,(u\,(A_{ss}\,u+2\,A_{zs}\,x+A_{ts}\,y)
              +A_{zz}\,x^2+2\,A_{tz}\,x\,y+A_{tt}\,y^2),
    \\
           &&
              A\,r_s+2\,A_s\,r
              +v_s\,(A\,u_s+A_s\,u+A_t\,y+A_z\,x\,v_s
              +A\,u_s\,v_s)
    \\
           &&
              +2\,v\,(2\,A_s\,u_s+A_{ss}\,u+A_{ts}\,y+A_{zs}\,x), 
    \\
    \phi_{0,1}(A) &\mapsto &(A\,u_s-A_s\,u-A_t\,y-A_z\,x, A\,v_s+2\,A_s\,v), 
    \\
    \phi_{1,0}(A) &\mapsto &(0, 
                             -2\,A\,p_s-A_s\,(p+2\,u\,u_s)
                             -u_s\,(A\,u_s+2\,A_t\,y+2\,A_z\,x)
    \\
           &&
              -\frac{1}{2}\,(u\,(A_{ss}\,u+2\,A_{ts}\,y+2\,A_{zs})
              +A_{zz}\,x^2+2\,A_{tz}\,x\,y+A_{tt}\,y^2)), 
    \\
    \phi_{1,1}(A) &\mapsto &(0, -2\,A\,u_s-A_s\,u-A_t\,y-A_z\,x).  
  \end{array}
\end{equation*}
The action of the recursion operator $\mathcal{R}_{0,0}$ is given by the first
components of the above formulas.

The action of the operator $\mathcal{R}_{0,0}^{-1}\circ \mathrm{pr}_u$, where
$\mathrm{pr}_u$ `forgets' the $v$-component of shadows, is presented in terms
of the nonlocal variables $w_0$ and $w_1$ as follows\footnote{Most of these
  formulas were obtained in \cite{BKMV2015}}:
\begin{equation*}
  \begin{array}{lcl}
    \psi_1 &\mapsto  & 
                       (x\,w_{0,x}+y\,w_{0,y})\,w_{0,s}^{-1}, 
    \\
    \psi_2 &\mapsto  &
                       t\,u_y+2\,z\,u_x-x\,w_{0,x}\,w_{0,s}^{-1}, 
    \\
    \psi_3 &\mapsto  &
                       t\,u_y-z\,u_x+(x\,w_{0,x}-y\,w_{0,y})\,w_{0,s}^{-1}, 
    \\
    \psi_4 &\mapsto  &
                       x\,(w_1-z\,w_{0,x}-t\,w_{0,y})
                       \,w_{0,s}^{-1}+z\,(z\,u_x+t\,u_y),   
    \\
    \psi_5 &\mapsto  &
                       t\,u_x-y\,w_{0,x}\,w_{0,s}^{-1}, 
    \\
    \psi_6 &\mapsto  &
                       -y\,(w_1-z\,w_{0,x}-y\,w_{0,y})
                       \,w_{0,s}^{-1}-t^2\,u_y-t\,z\,u_x,  
    \\
    \psi_7 &\mapsto  &
                       -u_x,
    \\
    \psi_8 &\mapsto  &
                       z\,u_y-x\,w_{0,y}\,w_{0,s}^{-1}, 
    \\
    \psi_9 &\mapsto  &
                       u_y,
    \\
    \psi_{10} &\mapsto  &
                          (w_1-z\,w_{0,x}-t\,w_{0,y})\,w_{0,s}^{-1}, 
    \\
    \psi_{11} &\mapsto  &
                          w_{0,x}\,w_{0,s}^{-1}, 
    \\
    \psi_{12} &\mapsto  &
                          -w_{0,y}\,w_{0,s}^{-1},
    \\
    \varphi_{0,0} (A) &\mapsto&  A,
    \\
    \varphi_{0,1} (A) &\mapsto&  0. 
  \end{array}
\end{equation*}

\subsection{Higher order Lax representations}\label{sec:higher-order-lax}

The symmetry $\psi_{13} =(0,-u_{sss})$ has no lift to the Lax representation
\eqref{extended_MASh5D_shadow_generated_covering}.  We can use this fact to
produce new Lax representations for Eq.~\eqref{MASh5D} and its cotangent
covering.  Indeed, the flow of the prolongation of the vector field
$u_{sss}\,\partial_v$ transforms System
\eqref{extended_MASh5D_shadow_generated_covering} to the system
\begin{equation}\label{higher_adjoint_linearized_MASh5D_covering}
  \begin{array}{lcl}
    \omega_t
    &=& \displaystyle{\Lambda\,\omega_y-u_y\,\omega_s
                 -\left(2\,u_{ys}+\frac{\kappa}{\kappa\,t
        +\mu\,z+1}\right)\,\omega} 
    \\
    &&
       \qquad
       +2\,(v_y+\,u_{ysss})\,\theta_s+(v_{ys}+u_{yssss})\,\theta,
    \\
    \omega_z &=& \displaystyle{\Lambda\,\omega_x-u_x\,\omega_s
                 -\left(2\,u_{xs}+\frac{\mu}{\kappa\,t+\mu\,z+1}\right)\,\omega}
    \\
    &&
       \qquad
       +2\,(v_x+\,u_{xsss})\,\theta_s+(v_{xs}+u_{xssss})\,\theta.
  \end{array}
\end{equation}
System \eqref{MASh5D_shadow_generated_covering},
\eqref{higher_adjoint_linearized_MASh5D_covering} defines a two-parametric
family of higher order Lax representations for System~\eqref{MASh5D},
\eqref{adjoint_linearized_MASh5D}. If we put $v=0$ in
\eqref{higher_adjoint_linearized_MASh5D_covering}, we obtain the system
\begin{equation}\label{higher_MASh5D_covering}
  \begin{array}{lcl}
    \omega_t
    &=& \displaystyle{\Lambda\,\omega_y-u_y\,\omega_s
        -\left(2\,u_{ys}+\frac{\kappa}{\kappa\,t+\mu\,z+1}\right)\,\omega}
    \\
    &&
       \qquad
       +2\,\,u_{ysss}\,\theta_s+\,u_{yssss}\,\theta,
    \\
    \omega_z
    &=& \displaystyle{\Lambda\,\omega_x-u_x\,\omega_s
        -\left(2\,u_{xs}+\frac{\mu}{\kappa\,t+\mu\,z+1}\right)\,\omega}
    \\
    &&
       \qquad
       +2\,\,u_{xsss}\,\theta_s+\,u_{xssss}\,\theta.
  \end{array}
\end{equation}
System \eqref{MASh5D_shadow_generated_covering},
\eqref{higher_MASh5D_covering} is compatible by virtue of Eq.~\eqref{MASh5D}
alone and defines a two-parametric family of higher order Lax representations
for this equation.

Likewise, the action of the flow of $u_{sss}\,\partial_v$ on System
\eqref{MASh5D_lift_of_main_covering} produces the systems
\begin{equation}\label{higher_extension_of_MASh5D_lift_of_main_covering_1}
  \begin{array}{lcl}
    W_t &=& \lambda\,W_y-u_y\,W_s -3 \,u_{ys} \,W+(v_y+u_{ysss})\,w_s,
    \\
    W_z &=& \lambda\,W_x-u_x\,W_s-3\,u_{xs}\,W+(v_x+u_{xsss})\,w_s
  \end{array}
\end{equation}
and
\begin{equation}\label{higher_extension_of_MASh5D_lift_of_main_covering_2}
  \begin{array}{lcl}
    W_t &=& \lambda\,W_y-u_y\,W_s -3 \,u_{ys} \,W+u_{ysss}\,w_s,
    \\
    W_z &=& \lambda\,W_x-u_x\,W_s-3\,u_{xs}\,W+u_{xsss}\,w_s.
  \end{array}
\end{equation}
These systems together with~\eqref{MASh5D_main_covering} give one-parametric
families of higher order Lax representations for System
\eqref{MASh5D}\textup{,} \eqref{adjoint_linearized_MASh5D} and
Eq. \eqref{MASh5D}, respectively.

\begin{remark}\label{sec:higher-order-lax-1}
  Generally speaking, higher symmetries, contrary to classical (contact) ones,
  cannot be used to insert a parameter to Lax pairs (coverings), since they do
  not possess trajectories. But in the case of~$\psi_{13}$ the result
  of~\cite{KrasilshchikVinogradov1989} remains valid, because the first
  component of the symmetry at hand vanished, while the second one is
  independent of~$u$. Of course, this is true for all symmetries of such a
  type.\hfill $\diamond$
\end{remark}

\section{Conclusions}\label{sec:concluding-remarks}

Let us conclude our exposition with the following remarks.
\begin{enumerate}
\item Since the cotangent equation $\tilde{\mathcal{E}} =
  \mathcal{T}^*\mathcal{E}$ is a Lagrangian one, one has
  $\ell_{\tilde{\mathcal{E}}} = \ell_{\tilde{\mathcal{E}}}^*$ and thus the
  spaces of symmetries and cosymmetries for~$\tilde{\mathcal{E}}$
  coincide. Therefore, the recursion operators found above are good candidates
  for Hamiltonian structures.
\item Contrary to many other examples, the 5D Mart{\'{\i}}nez Alonso--Shabat
  equation admits a rich family of recursion operators. It would interesting
  to describe the group structure of this family.
\item It is generally accepted that nonlinear multi-dimensional equations do
  not possess higher symmetries. Our experience shows that cotangent equations
  deliver a `regular' counter-example to this statement. What is the reason of
  this phenomenon and what is the role of higher symmetries in geometry of
  multi-dimensional systems?
\item As it was noticed in Section~\ref{sec:introduction}, a number of Lax
  integrable systems are obtained as symmetry reduction of our equation or are
  related to it by B\"{a}cklund transformations. What is the behavior of
  various invariants (symmetry algebras, recursion operators, Lax pairs, etc.)
  under these reduction and/or relations?
\item Finally, it is interesting to compare the obtained results with the
  invariants of other  multi-dimensional equations.
\end{enumerate}
These and other problems are subjects of future research. 

\section*{Acknowledgments}
Computations were supported by the {\sc Jets} software, \cite{Jets}.

\end{document}